\documentclass[a4paper,11pt]{article}
\usepackage[utf8]{inputenc}
\usepackage[margin=3cm]{geometry}
\usepackage{graphicx}
\usepackage{amsmath, amssymb, amsthm}
\usepackage[pdfpagelabels, pagebackref, naturalnames]{hyperref}
\usepackage[font=small, labelfont=bf]{caption}
\usepackage[labelfont=default]{subcaption}
\usepackage{xspace}
\usepackage{wrapfig}
\usepackage{booktabs}
\usepackage{enumitem}

\makeatletter
\renewcommand*{\@fnsymbol}[1]{\ifcase#1\or*\else\@arabic{\numexpr#1-1\relax}\fi}
\makeatother

\theoremstyle{plain}

\newtheorem{theorem}{Theorem} 
\newtheorem{lemma}{Lemma}
\theoremstyle{definition}  

\newcommand{\WLOG}{w.l.o.g.\xspace}

\newcommand{\prob}[1]{\textsc{#1}}
\newcommand{\crown}{\prob{Max}-\prob{Crown}\xspace}
\newcommand{\gap}{\prob{Gap}\xspace}
\newcommand{\eps}{\ensuremath{\varepsilon}}
\DeclareMathOperator{\opt}{OPT}
\DeclareMathOperator{\alg}{ALG}
\DeclareMathOperator{\GAP}{GAP}

\graphicspath{{images/}}

\begin{document}

\title{\bf Improved Approximation Algorithms \\
  for Box Contact Representations\thanks{A preliminary version of
	this paper has appeared in \emph{Proc. 22nd Eur. Symp. Algorithms
	(ESA'14)}, volume 8737 of \emph{Lect. Notes Comput. Sci.},
	pages 87--99, Springer-Verlag.
	Ph.~Kindermann and A.~Wolff acknowledge support by the ESF EuroGIGA
	project GraDR. S.~Kobourov and S.~Pupyrev are supported by NSF grants
	CCF-1115971 and DEB 1053573}}

\author{Michael A. Bekos\thanks{Wilhelm-Schickard-Institut f\"{u}r
    Informatik, Universit\"{a}t T\"{u}bingen, Germany.}
  \and
  Thomas C. van Dijk\thanks{Lehrstuhl f\"ur Informatik I,
    Universit\"at W\"urzburg, Germany.
    }
  \and
  Martin Fink\footnotemark[3]$~^,$\thanks{Department of Computer Science,
		University of California, Santa Barbara USA.}
  \and
  Philipp~Kindermann\footnotemark[3]
  \and
  Stephen~Kobourov\thanks{Department of Computer Science,
    University of Arizona, USA.}
  \and 
  Sergey~Pupyrev\footnotemark[5]$~^,$\thanks{Institute of Mathematics
	and Computer Science, Ural Federal University, Russia.}
  \and 
  Joachim~Spoerhase\footnotemark[3]
  \and 
  Alexander~Wolff\footnotemark[3]
}

\maketitle

\begin{abstract}
  We study the following geometric representation problem: Given a
  graph whose vertices correspond to axis-aligned rectangles with
  fixed dimensions, arrange the rectangles without overlaps in the
  plane such that two rectangles touch if the graph
  contains an edge between them.  This problem is called
  \prob{Contact Representation of Word Networks} (\prob{Crown}) since
  it formalizes the geometric problem
  behind drawing word clouds in which semantically related words are
  close to each other.  \prob{Crown} is known to be
  NP-hard, and there are approximation algorithms for certain graph
  classes for the optimization version, \crown, in which realizing
  each desired adjacency yields a certain profit.

  We present the first $O(1)$-approximation algorithm for the general
  case, when the input is a complete weighted graph, and for the
  bipartite case.  Since the
  subgraph of realized adjacencies is necessarily planar, we also consider
  several planar graph classes (namely stars, trees, outerplanar, and
  planar graphs), improving upon the known results.
  For some graph classes, we also describe improvements
  in the unweighted case, where each adjacency yields the same
  profit. Finally, we show that the problem is APX-complete on
  bipartite graphs of bounded maximum degree.
\end{abstract}

\section{Introduction}

\begin{figure}[!t]
    \center
    \includegraphics[height=5.3cm]{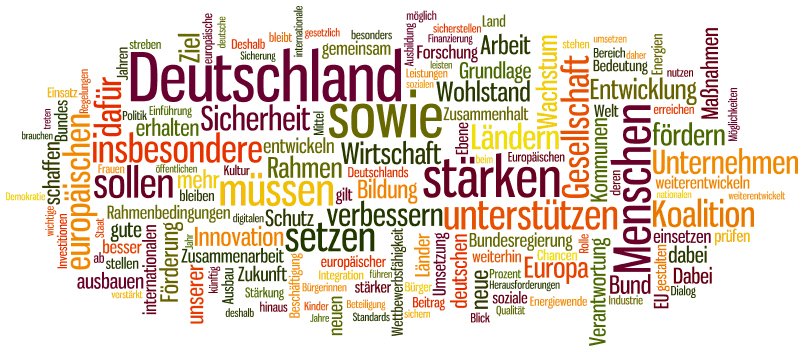}
   \caption{Der Koalitionsvertrag im Schnellcheck (Quick
     overview of the [new German] coalition agreement), Spiegel
     Online, Nov. 27, 2013,~\cite{spiegel2013}
		 (Click  on ``Fotos''.)}
\label{fig:spiegel}
\end{figure}

In the last few years, word clouds have become a standard tool for
abstracting, visualizing, and comparing text documents.  For example,
word clouds were used in 2008 to contrast the speeches of the US
presidential candidates Obama and McCain.  More recently, the German
media used them to visualize the newly signed coalition agreement and
to compare it to a similar agreement from 2009; see Fig.~\ref{fig:spiegel}.
A word cloud of a
given document consists of the most important (or most frequent)
words in that document.  Each word is printed in a given font and
scaled by a factor roughly proportional to its importance (the same is
done with the names of towns and cities on geographic maps, for
example).  The printed words are arranged without overlap and tightly
packed into some shape (usually a rectangle).  Tag clouds look
similar; they consist of keyword metadata (tags) that have been
attributed to resources in some collection such as web pages or
photos.

Wordle~\cite{viegas2009} is a popular tool for drawing word or tag clouds.
The Wordle website allows users to
upload a list of words and, for each word, its relative importance.
The user can further select font, color scheme, and decide whether all
words must be placed horizontally or whether words can also be placed
vertically.  The tool then computes a placement of the words, each
scaled according to its importance, such that no two words overlap.
Generally, the drawings are very compact and aesthetically appealing.

In the automated analysis of text one is usually not just
interested in the most important words and their frequencies, but also
in the connections between these words.  For example, if a pair of words often
appears together in a sentence, then this is often seen as evidence
that this pair of words is linked semantically~\cite{l-wcdbc-JNLE02}.
In this case,
it makes sense to place the two words close to each other in the word
cloud that visualizes the given text.
This is captured by an input graph $G=(V,E)$ of
desired contacts.  We are also given, for each vertex $v \in V$, the
dimensions (but not the position) of a \emph{box}~$B_v$, that is, an
axis-aligned rectangle.  We denote the height and width of~$B_v$
by~$h(B_v)$ and~$w(B_v)$, respectively, or, more briefly, by $h(v)$
and~$w(v)$.  For each edge $e=(u,v)$ of~$G$, we are given a positive
number $p(e)=p(u,v)$, that corresponds to the \emph{profit} of~$e$.
For ease of notation,
we set $p(u,v)=0$ for any non-edge $(u,v) \in V^2 \setminus E$ of $G$.

Given a box~$B$ and a point~$q$ in the plane, let~$B(q)$ be a
placement of $B$ with lower left corner~$q$.  A \emph{representation}
of~$G$ is a map~$\lambda \colon V \to \mathbb{R}^2$ such that for any
two vertices~$u \ne v$, it holds that~$B_u(\lambda(u))$
and~$B_v(\lambda(v))$ are interior-disjoint.  Boxes may \emph{touch},
that is, their boundaries may intersect.  If the intersection is
non-degenerate, that is, a line segment of positive length, we say
that the boxes are \emph{in contact}.  We say that a
representation~$\lambda$ \emph{realizes}~an edge $(u,v)$ of~$G$ if
boxes~$B_u(\lambda(u))$ and~$B_v(\lambda(v))$ are in contact.

This yields the problem \emph{Contact Representation of Word Networks}
(\textsc{Crown}):
Given an edge-weighted graph~$G$ whose vertices correspond to
boxes, find a representation of~$G$ with the vertex boxes such that
every edge of~$G$ is realized.  In this paper, we study the
optimization version of \textsc{Crown}, \crown, where  the aim is to
maximize the total profit (that is, the sum of the weights) of the
realized edges.
We also consider the unweighted version of the problem, where all
desired contacts yield a profit of~1.

\paragraph{Previous Work.}

Barth et al.~\cite{bfklnopsuw-swcrh-LATIN14} recently introduced
\crown and showed that the problem is strongly NP-hard even for trees
and weakly NP-hard even for stars.  They presented an exact algorithm
for cycles and approximation algorithms for stars, trees,
planar graphs, and graphs of constant maximum degree; see the first
column of Table~\ref{table:res}.  Some of their
solutions use an approximation
algorithm with ratio $\alpha = e/(e-1) \approx 1.58$~\cite{Fleischer2011}
for the \prob{Generalized Assignment Problem} (\gap), defined as
follows: Given a set of bins with capacity constraints and a
set of items that possibly have different sizes and values for each
bin, pack a maximum-valued subset of items into the bins.  The problem
is APX-hard~\cite{mkpptas}.

\crown is related to finding \emph{rectangle representations} of
graphs, where vertices are represented by
axis-aligned rectangles with non-intersecting interiors and edges
correspond to rectangles with a common boundary of non-zero length.
Every graph that can be represented this way is planar and every
triangle in such a graph is a facial triangle.  These two conditions
are also sufficient to guarantee a rectangle representation~\cite{buchsbaum08}.
Rectangle representations play an important role in VLSI layout, cartography,
and architecture (floor planning). In a recent survey,
Felsner~\cite{felsner2013rectangle} reviews many rectangulation
variants. Several interesting problems arise when the
rectangles in the representation are restricted.
Eppstein et al.~\cite{eppstein2012area}
consider rectangle representations which can realize any given area-requirement
on the rectangles, so-called \emph{area-preserving rectangular
  cartograms}, which were introduced by Raisz~\cite{r-rsc-34} already
in the 1930s. Unlike cartograms, in our setting there is no
inherent geography, and hence, words can be positioned
anywhere.  Moreover, each word has fixed dimensions enforced by its
importance in the input text, rather than just fixed area.
N\"ollenburg et al.~\cite{nollenburg2013edge} recently considered a
variant where the edge weights prescribe the length of the desired
contacts.

Finally, the problem of computing semantics-aware word clouds
is related to classic graph layout problems, where the goal is to draw graphs so
that vertex labels are readable and Euclidean distances between
pairs of vertices are proportional to the underlying graph distance
between them. Typically, however, vertices are treated as points
and label overlap removal is a post-processing step~\cite{dwyer05,gh10}.
Most tag cloud and word cloud tools such as Wordle~\cite{viegas2009}
do not show the semantic relationships between words, but
force-directed graph layout heuristics are sometimes used to add such
functionality~\cite{swcTR,cui2010context,semantic_collections,wu2011semantic}.
For an example output of such a tool, see Fig.~\ref{fig:example}.

\begin{figure}[b]
  \centering
  \fbox{\includegraphics[width=.8\textwidth]{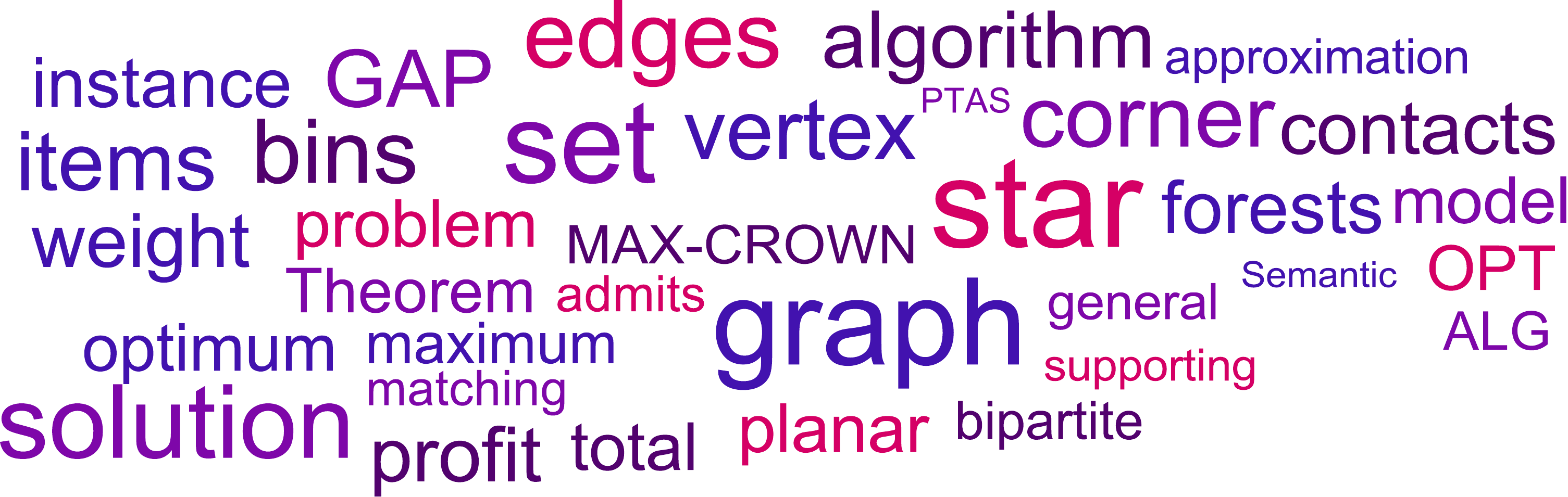}}
  \caption{Semantics-preserving word cloud for the 35 most
    ``important'' words in this paper.  Following the text processing
    pipeline of Barth et al.~\cite{swcTR}, these are the words ranked
    highest by LexRank~\cite{Erkan04}, after removal of stop words
    such as ``the''.  The edge profits are proportional to
    the relative frequency with which the words occur in the same
    sentences.  The layout algorithm of Barth et al.~\cite{swcTR}
    first extracts a heavy star forest from the weighted input graph as in
    Theorem~\ref{thm:approx-gen-graph-weighted} and then applies a
    force-directed post-processing.}
  \label{fig:example}
\end{figure}


\paragraph{Model.}
We consider two different models. In Sections~\ref{sub:weighted} 
and~\ref{sub:unweighted}, we do not count
\emph{point contacts}, that is, we consider two boxes in
contact only if their intersection is a line segment of positive
length. Hence, the contact graph of the boxes
is planar.  This model is used in most work on rectangle contact representations.
In Section~\ref{sec:point-model}, we describe how to
modify our algorithms to guarantee $O(1)$-approximations also in the model
that allows and rewards point contacts.
We allow words only to be placed horizontally.

\paragraph{Our Contribution.}

Known results and our contributions to \crown are shown in
Table~\ref{table:res}.
\begin{table}[tb]
  \centering
  \caption{Previously known and new results for the unweighted
    and weighted versions of \crown (for $\alpha \approx 1.58$ and any
    $\eps>0$). The exact approximation factors are denoted in the
		corresponding theorems.}
  \label{table:res}
  \medskip

  \begin{tabular}{l@{~~}cc@{~~~}c@{\qquad}c@{~~~}c}
    \toprule
    & \multicolumn{3}{c}{Weighted} &
    \multicolumn{2}{c}{Unweighted} \\
    \cmidrule(l{1ex}r{5ex}){2-4}
    \cmidrule(l{1ex}r{1ex}){5-6}
    Graph class & Ratio~\cite{bfklnopsuw-swcrh-LATIN14} & Ratio [new]
    & Ref. & Ratio & Ref.\\
    \midrule 
    cycle, path & $1$ & & & & \\ 
    star & $\alpha$ 
    & $1 + \eps$ & Thm.~\ref{thm:previous-improved} & & \\
    tree & $2 \alpha$ 
    & $2 + \eps$ & Thm.~\ref{thm:previous-improved}
    & $2$ & Thm.~\ref{thm:2-approx-tree-unweighted} \\
    & NP-hard \\
    max-degree $\Delta$ & $\lfloor(\Delta + 1) / 2\rfloor$ \\
    planar max-deg.~$\Delta$ & & & & $1+\eps$ &
    Thm.~\ref{thm:ptas-bounded-deg-planar} \\
    outerplanar & & $3+\eps$ &
    Thm.~\ref{thm:approx-outerplanar-weighted} & & \\
    planar & $5 \alpha$ 
    & $5 + \eps$ & Thm.~\ref{thm:previous-improved} & & \\
    bipartite & & APX-complete & Thm.~\ref{thm:bipartite-weighted-apx-complete} & & \\
		~~~without point contacts& & $\approx 8.4$ &
    Thm.~\ref{thm:approx-bipartite-weighted} & & \\
		~~~with point contacts& & $\approx 9.5$ &
    Thm.~\ref{thm:point-contacts-weighted} & & \\
    general \\
		~~~without point contacts & & $\approx 16.9$ (rand.) &
    Thm.~\ref{thm:randomized-approx-gen-graph-weighted}
    & $\approx 13.4$ & Thm.~\ref{thm:approx-gen-graph-unweighted}\\
    & & $\approx 21.1$ (det.) &
    Thm.~\ref{thm:approx-gen-graph-weighted} & & \\
		~~~with point contacts & & $\approx 19$ (rand.) &
    Thm.~\ref{thm:point-contacts-weighted}
    & $\approx 16.5$ & Thm.~\ref{thm:point-contacts-unweighted}\\
    & & $\approx 22.1$ (det.) &
    Thm.~\ref{thm:point-contacts-weighted} & & \\
    \bottomrule
  \end{tabular}
\end{table}
Note that the results of Barth et al.\ \cite{bfklnopsuw-swcrh-LATIN14}
in column~1 are simply based on existing decompositions of the
respective graph classes into star forests or cycles.

Our results rely on a variety of algorithmic tools.  First, we devise
sophisticated decompositions of the input graphs into heterogeneous
classes of subgraphs, which also requires a more general combination
method than that of Barth et al. Second, we use randomization to
obtain a simple constant-factor approximation for general weighted
graphs.  Previously, such a result was not even known for unweighted
bipartite graphs.  Third, to obtain an improved algorithm for the unweighted
case, we prove a lower bound on the size of a matching in a planar
graph of high average degree.  Fourth, we use a planar separator
result of Frederickson \cite{frederickson87shortestpaths-planar} to
obtain a polynomial-time approximation scheme (PTAS) for
degree-bounded planar graphs.

Our other main result is the use of the combination lemma, which,
among others, yielded the first approximation algorithms for bipartite
and for general graphs; see Section~\ref{sub:weighted}.
For general graphs, we present a simple
randomized solution (based on the solution for bipartite graphs) and
a more involved deterministic algorithm.  For trees, planar graphs of
constant maximum degree, and general graphs, we have improved results
in the unweighted case; see Section~\ref{sub:unweighted}.  For the
model with point contacts, we show how to adjust the approximation algorithms for
bipartite and general graphs; see Section~\ref{sec:point-model}.
Finally, we show APX-completeness for bipartite graphs of maximum degree~9 (see
Section~\ref{sub:hardness}) and list some open problems (see
Section~\ref{sec:open}).

\paragraph{Runtimes.} Most of our algorithms involve approximating a
number of \gap instances as a subroutine, using either the PTAS
\cite{bkv-atumd-SICOMP11} if the number of bins is constant or the
approximation algorithm of Fleischer et~al.~\cite{Fleischer2011} for
general instances. Because of this, the runtime of
our algorithms consists mostly of approximating \gap instances.
Both algorithms to approximate \gap instances solve linear programs, so we
refrain from explicitly stating the runtime of these algorithms.

For practical purposes, one can use a purely combinatorial approach
for approximating \gap~\cite{CKR06}, which utilizes an algorithm for
the \prob{Knapsack} problem as a subroutine. The algorithm translates into
a $3$-approximation for \gap running in $O(NM)$ time (or a
$(2+\eps)$-approximation running in $O(MN\log1/\eps + M/\eps^4)$ time), where $N$ is the number
of items and $M$ is the number of bins. In our setting, the simple $3$-approximation implies
a randomized $32$-approximation (or a deterministic $40$-approximation) algorithm with running time
$O(|V|^2)$ for \crown on general weighted graphs.


\section{Some Basic Results}
\label{sec:preliminaries}

In this section, we present two technical lemmas that will help us to
prove our main results in the following two sections where we treat
the weighted and unweighted cases of \crown.  The second lemma
immediately improves the results of Barth et
al.~\cite{bfklnopsuw-swcrh-LATIN14} for stars, trees, and
planar graphs.

\subsection{A Combination Lemma}

Several of our algorithms cover the input graph with subgraphs that
belong to graph classes for which the \crown problem is known to admit
good approximations.  The following lemma allows us to
combine the solutions for the subgraphs.  We say that a graph $G=(V,
E)$ is \emph{covered} by graphs $G_1=(V,E_1), \dots, G_k=(V,E_k)$ if
$E = E_1 \cup \dots \cup E_k$.

\begin{lemma}
  \label{lem:partitioned-approximation}
  Let graph $G=(V,E)$ be covered by graphs~$G_1,G_2,\dots,G_k$.  If,
  for $i=1,2,\ldots,k$, weighted \crown on graph~$G_i$ admits an
  $\alpha_i$-approximation, then weighted \crown on $G$ admits a
  $\left(\sum_{i=1}^k \alpha_i\right)$-approximation.
\end{lemma}
\begin{proof}
  Our algorithm works as follows.  For $i=1,\dots,k$, we apply the
  $\alpha_i$-approximation algorithm to~$G_i$ and report the result
  with the largest profit as the result for~$G$.  We show that this
  algorithm has the claimed performance guarantee.  For the
  graphs $G, G_1, \dots, G_k$, let $\opt,\opt_1,\dots,\opt_k$ be the
  optimum profits and let $\alg,\alg_1,\dots,\alg_k$ be the profits of the
  approximate solutions.  By definition, $\alg_i\geq\opt_i/\alpha_i$ for
  $i=1,\dots,k$.  Moreover, $\opt\leq\sum_{i=1}^k\opt_i$ because the
  edges of~$G$ are covered by the edges of $G_1,\dots,G_k$.  Assume,
  \WLOG, that $\opt_1/\alpha_1=\max_i (\opt_i/\alpha_i)$.  Then
  \begin{displaymath}
    \alg \;=\; \alg_1 \;\ge\; \frac{\opt_1}{\alpha_1} \;\ge\;
    \frac{\sum_{i=1}^k\opt_i}{\sum_{i=1}^k\alpha_i} \;\ge\;
    \frac{\opt}{\sum_{i=1}^k\alpha_i}\,.  \qedhere
  \end{displaymath}
\end{proof}

\subsection{Improvement on existing approximation algorithms}
\label{sub:gap-ptas}

The approximation algorithms for stars, trees and planar graphs 
provided by Bekos et al.~\cite{bfklnopsuw-swcrh-LATIN14} use
an $\alpha$-approximation algorithm for \\GAP instances.
We prove that these instances require only a constant number of
bins and thus can be approximated using the PTAS of
Briest et al.~\cite{briest05apx-mechanism-design}.

\begin{lemma}[\cite{bkv-atumd-SICOMP11}]
  \label{lem:gap-ptas}
  For any $\epsilon>0$, there is a $(1+\epsilon)$-approximation
  algorithm for \gap with a constant number of bins. The algorithm takes $n^{O(1/\epsilon)}$ time.\qed{}
\end{lemma}
Using Lemmas~\ref{lem:partitioned-approximation}
and~\ref{lem:gap-ptas}, we improve the approximation algorithms of
Barth et al.~\cite{bfklnopsuw-swcrh-LATIN14}.
\begin{theorem}
  \label{thm:previous-improved}
  Weighted \crown admits
  a~$(1 + \eps)$-approximation algorithm on stars,
  a~$(2 + \eps)$-approximation algorithm on trees,
  and a~$(5 + \eps)$-approximation algorithm on planar graphs.
\end{theorem}
\begin{proofWithWrapfig}
  By Lemma~\ref{lem:partitioned-approximation}, the claim for stars
  implies the other two claims since a tree can be
  covered by two star forests and a planar graph can be
  covered by five star forests in polynomial time~\cite{hakimi96}.
  
  \begin{wrapfigure}{r}{4cm} \vspace{-2ex}
    \centering
    \includegraphics{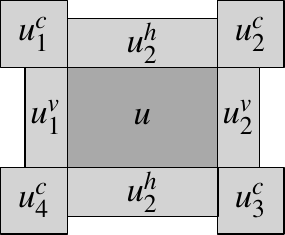}
    \caption{Notation for the PTAS for stars}
    \label{fig:star-ptas}
  \end{wrapfigure}

  We now show that we can use Lemma~\ref{lem:gap-ptas} to get a PTAS
  for stars.  First, we give the PTAS for the model with point contacts.

  Let $u$ be the center vertex of the star.  We create eight bins: four
  \emph{corner bins} $u_1^c, u_2^c, u_3^c$, and $u_4^c$ modeling
  adjacencies on the four corners of the box $u$, two \emph{horizontal
  bins}~$u_1^h$ and~$u_2^h$ modeling adjacencies on the top and bottom
  side of~$u$, and two \emph{vertical bins}~$u_1^v$ and~$u_2^v$ modeling
  adjacencies on the left and right side of~$u$; see
  Fig.~\ref{fig:star-ptas}.
	The capacity of the corner bins is 1, the capacity of the 
	horizontal bins is the width $w(u)$ of $u$, and the capacity of the 
	vertical bins is the height $ h(u)$ of $u$.  Next, we introduce an 
	item~$i(v)$ for any leaf vertex~$v$ of the star.  The size of $i(v)$
	is~1 in any corner bin, $w(v)$ in any horizontal bin, and $h(v)$ 
	in any vertical bin.  The profit of~$i(v)$ in any bin is the
	profit~$p(u,v)$ of the edge~$(u,v)$.
  
  Note that any feasible solution to the \crown instance can be
  normalized so that any box that touches a corner of $u$ has a
  point contact with $u$.  Hence, the above is an
  approximation-preserving reduction from weighted \crown on stars
  (with point contacts) to \gap.  By Lemma~\ref{lem:gap-ptas}, we
  obtain a PTAS.
    
	We first assume that all boxes have integral edge lengths, which 
	can be accomplished by scaling.  Consider a feasible solution 
	without point contacts.  We now modify the solution as follows.  
	Each box that touches a corner of $u$ is moved so that it has a 
	point contact with this corner.  Afterwards, we move some of the 
	remaining boxes until all corners of $u$ have point contacts or 
	until we run out of boxes.  This yields a solution with point 
	contacts in which there are two opposite sides of $u$---say the two 
	horizontal sides---which either do not touch any box or from which 
	we removed one box during the modification.  Now observe that, if 
	we shrink the two horizontal sides by an amount of $1/2$, then all 
	contacts can be preserved since there was a slack of at least $1$ 
	at both horizontal sides.  Conversely, observe that any feasible 
	solution with point contacts to the modified instance with shrunken 
	horizontal sides can be transformed into a solution without point 
	contacts since we always have a slack of at least $1/2$ on both 
	horizontal sides.  This shows that there is a correspondence 
	between feasible solutions without point contacts and feasible 
	solutions with point contacts to a modified instance where we 
	either shrink the horizontal or the vertical sides by $1/2$.  The 
	PTAS for \crown on stars consists in applying a PTAS to two 
	instances of \crown with point contacts where we shrink the 
	horizontal or vertical sides, respectively, and in outputting the 
	better of the two solutions.
 \end{proofWithWrapfig}

\section{The Weighted Case}
\label{sub:weighted}

In this section, we provide new approximation algorithms for more
involved classes of (weighted) graphs than in the previous section.
Recall that~$\alpha = e/(e-1) \approx 1.58$.
First, we give a $(3 + \eps)$-approximation for outerplanar graphs.
Then, we present a $16\alpha/3$-approximation for bipartite graphs.
For general graphs, we provide a simple randomized
$32\alpha/3$-approximation and a deterministic
$40\alpha/3$-approximation.

\begin{theorem}
  \label{thm:approx-outerplanar-weighted}
  Weighted \crown on outerplanar graphs admits a $(3+\eps)$-ap\-prox\-imation.
\end{theorem}
\begin{proof}
  It is known that the star arboricity of an outerplanar graph is~3,
  that is, it can be partitioned into at most three star
  forests~\cite{hakimi96}.  Here we give a simple algorithm for
  finding such a partitioning.

  Any outerplanar graph has degeneracy at most~2, that is, it has a vertex of
  degree at most~2.  We prove that any outerplanar graph~$G$ can be
  partitioned into three star forests such that every vertex of~$G$ is the
  center of only one star.  Clearly, it is sufficient to prove the
  claim for maximal outerplanar graphs in which all vertices have
  degree at least~2.
  We use induction on the number of vertices of~$G$.  The base of the
  induction corresponds to a $3$-cycle for which the claim clearly holds.
  For the induction step, let $v$ be a degree-2 vertex of~$G$ and
  let~$(v, u)$ and~$(v, w)$ be its incident edges.  The graph~$G - v$
  is maximal outerplanar and thus, by induction hypothesis, it can be
  partitioned into star forests~$F_1$,~$F_2$, and~$F_3$ such that~$u$ is
  the center of a star in~$F_1$ and~$w$ is the center of a star in~$F_2$.
  Now we can cover~$G$ with three star forests: we add~$(v, u)$
  to~$F_1$, we add $(v, w)$ to~$F_2$, and we create a new star
  centered at~$v$ in~$F_3$.

  Applying Lemma~\ref{lem:partitioned-approximation}
  and Theorem~\ref{thm:previous-improved} to the star forests
  completes the proof.
\end{proof}

\begin{theorem}
  \label{thm:approx-bipartite-weighted}
  Weighted \crown on bipartite graphs admits a $16\alpha/3(\approx8.4)$-ap\-prox\-imation.
\end{theorem}
\begin{proof}
  Let~$G = (V, E)$ be a bipartite input graph
  with~$V = V_1 \;\dot{\cup}\; V_2$ and~$E \subseteq V_1 \times V_2$.
  Using~$G$, we build an instance of \gap as follows. For each
  vertex~$u \in V_1$, we create eight
  bins~$u_1^c,u_2^c, u_3^c,u_4^c,u_1^h,u_{2}^h,u_1^v,u_{2}^v$ and
  set the capacities exactly as we did for the star center in
  Theorem~\ref{thm:previous-improved}. Next, we add an
  item~$i(v)$ for every vertex~$v \in V_2$.  The size of~$i(v)$ is,
  again, $1$ in any corner bin, $w(v)$ in any horizontal bin, and
  $h(v)$ in any vertical bin.  For $u \in V_{1}$, the profit of~$i(v)$
  is~$p(u,v)$ in any bin of~$u$.

  It is easy to see that solutions to the \gap instance are
  equivalent to word cloud solutions (with point contacts) in which
  the realized edges correspond to a forest of stars with all star
  centers being vertices of~$V_1$. Hence, we can find an approximate solution of
  profit~$\alg_1' \ge \opt_1'/\alpha$ where~$\opt_1'$ is the profit of an
  optimum solution (with point contacts) consisting of a star forest
  with centers in~$V_1$.

  We now show how to get a solution without point contacts.
  If the three bins on the top side of a vertex~$u$ (two corner bins and
	one horizontal bin) are not completely full, we can slightly move
  the boxes in the corners so that point contacts are avoided.
  Otherwise, we remove the lightest item from one of these bins.
  We treat the three bottommost bins analogously.
  Note that in both cases
  we only remove an item if all three bins are completely full.
  The resulting solution can be realized without point contacts. We do
  the same for the three left and three right bins and 
  choose the heavier of the two solutions.
  It is easy to see that we lose at most~$1/4$
  of the profit for the star center~$u$: Assume that the heaviest
  solution results from removing weight $w_1$ from one of the upper
  and weight $w_2$ from one of the lower bins. As we remove the
  lightest items only, the remaining weight from the upper and lower
  bins is at least $2(w_1 + w_2)$. On the other hand, the weight in
  the two vertical at least $w_1+w_2$; otherwise, dropping everything
  from these vertical bins would be cheaper. Hence, we keep at least
  weight $3(w_1+w_2)$.

	If we do so for all star centers, we get a solution with
	profit~$\alg_1 \ge 3/4 \cdot \alg_1' \ge 3\opt_1'/(4\alpha) \ge
	3\opt_1/(4\alpha)$ where~$\opt_1$ is the profit of an optimum
	solution (without point contacts) consisting of a star forest
	with centers in~$V_1$.
 
  Similarly, we can find a solution of profit~$\alg_2 \ge
  3\opt_2/(4\alpha)$ with star centers in~$V_2$, where~$\opt_2$ is the
  maximum profit that a star forest with centers in~$V_2$ can realize.
  Among the two solutions, we pick the one with larger profit~$\alg = \max
  \left\{ \alg_1,\alg_2 \right\}$.

  \begin{figure}[t]
  \begin{subfigure}[t]{.29\textwidth}
    \centering
    \includegraphics[page=1]{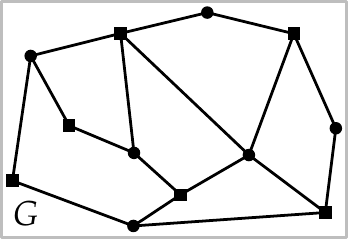}
    \caption{The graph $G^\star$ realized by an optimum solution is
    planar and bipartite.}
    \label{fig:bipartite-decomposition1}
  \end{subfigure}
  \hfill
  \begin{subfigure}[t]{.64\textwidth}
    \centering
    \includegraphics[page=2]{bipartite}
    \hfill
    \includegraphics[page=3]{bipartite}
    \caption{$G^\star$ can be decomposed into two forests $H_1$ and
    $H_2$ and further into four star forests $S_1, S_2$ (black) with
    centers in $V_1$ (disks) and $S_1',S_2'$ (dashed) with centers in
    $V_2$ (boxes).}
    \label{fig:bipartite-decomposition2}
  \end{subfigure}
  \caption{Partitioning the optimum solution in the
    proof of Theorem~\ref{thm:approx-bipartite-weighted}}
  \label{fig:bipartite-decomposition}
\end{figure}

  Let~$G^\star = (V,E^\star)$ be the contact graph realized by a fixed
  optimum solution, and let $\opt=p(E^\star)$ be its total profit.
  We now show that~$\alg \ge 3\opt/(16\alpha)$.
  As~$G^\star$ is a planar bipartite graph,~$|E^\star| \le 2n-4$.
  Hence, we can decompose~$E^\star$ into two forests~$H_1$ and~$H_2$ using
  a result of Nash-Williams~\cite{n-dfgf-JLMS64}; see
  Fig.~\ref{fig:bipartite-decomposition}.
  We can further decompose~$H_1$ into two star forests~$S_1$
  and~$S_{1}'$ in such a way that the star
  centers of~$S_1$ are in~$V_1$ and the star centers of~$S_{1}'$ are
  in~$V_2$. Similarly, we decompose~$H_2$ into a forest~$S_2$ of stars with
  centers in~$V_1$ and a forest~$S_{2}'$ of stars with centers in~$V_2$.
  As we decomposed the optimum solution into four star forests, one of
  them---say~$S_1$---has profit $p(S_1) \ge \opt/4$.  On the
  other hand,~$\opt_1 \ge p(S_1)$. Summing up, we get
  \[
    \alg \;\ge\; \alg_1 \;\ge\; 3\opt_1/(4\alpha) \;\ge\;
    3p(S_1)/(4\alpha) \;\ge\; 3\opt/(16\alpha). \qedhere
  \]
\end{proof}

\begin{theorem}
  \label{thm:randomized-approx-gen-graph-weighted}
  Weighted \crown on general graphs admits a randomized $32\alpha/3(\approx16.9)$-ap\-prox\-imation.
\end{theorem}
\begin{proof}
  Let~$G = (V,E)$ be the input graph and let~$\opt$ be the weight of a fixed optimum solution.
  Our algorithm works as follows. We first randomly partition the set of
  vertices into~$V_1$ and~$V_2 = V \setminus V_1$, that is, the
  probability that a vertex~$v$ is included in~$V_1$ is~$1/2$. Now we
  consider the bipartite graph~$G' = (V_1 \;\dot{\cup}\; V_2, E')$
  with~$E' = \left\{ (v_1,v_2) \in E \mid v_1 \in V_1
  \text{ and } v_2 \in V_2 \right\}$ that is induced by~$V_1$ and~$V_2$.
  By applying Theorem~\ref{thm:approx-bipartite-weighted} on~$G'$, we can find a
  feasible solution for~$G$ with weight~$\alg \ge 3\opt'/(16\alpha)$,
  where~$\opt'$ is the weight of an optimum solution for~$G'$.

  Any edge of the optimum solution is contained in~$G'$ with
  probability~$1/2$. Let~$\overline{\opt}$ be the total weight of
  the edges of the optimum solution that are present in~$G'$.
  Then,~$E[\overline{\opt}] = \opt/2$. Hence,
  \[ E[\alg] \;\ge\; 3 E[\opt']/(16\alpha) \;\ge\; 3
  E[\overline{\opt}]/(16\alpha) \;=\; 3 \opt/(32\alpha). \qedhere
  \] 
\end{proof}

\begin{theorem}
  \label{thm:approx-gen-graph-weighted}
  Weighted \crown on general graphs admits a $40\alpha/3(\approx21.1)$-ap\-prox\-imation.
\end{theorem}
\begin{proofWithWrapfig}
Let~$G=(V,E)$ be the input graph.  As in the proof of
Theorem~\ref{thm:approx-bipartite-weighted},
our algorithm constructs an instance of \gap based on~$G$. The
difference is that, \emph{for every vertex} $v\in V$, we create
\emph{both eight bins and an item}~$i(v)$.  Capacities and sizes
remain as before.  The profit of placing item~$i(v)$ in a bin
of vertex~$u$, with $u \ne v$, is~$p(u,v)$.

Let~$\opt$ be the value of an optimum solution of \crown in~$G$, and
let $\opt_{\GAP}$ be the value of an optimum solution for the constructed
instance of \gap. Since any optimum solution of \crown, being a
planar graph, can be decomposed into five star forests~\cite{hakimi96},
there exists a star forest carrying at least~$\opt/5$ of the total
profit. Such a star forest corresponds to a solution
of \gap for the constructed instance; therefore,~$\opt_{\GAP} \ge \opt/5$.
Now we compute an $\alpha$-approximation for the \gap instance, which
results in a solution of total profit
$\alg_{\GAP} \ge \opt_{\GAP}/\alpha \ge \opt/(5\alpha)$.
Next, we show how our 
solution induces a feasible solution of \crown where every
vertex~$v\in V$ is either a bin or an item.

\begin{wrapfigure}{r}{4.5cm} 
  \centering
  \includegraphics{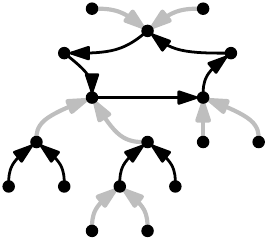}
  \caption{Partitioning a 1-tree into a star forest (gray) and the
    union of a cycle and a star forest (black)}
  \label{fig:1tree}
\end{wrapfigure}

Consider the directed graph~$G_{\GAP}=(V,E_{\GAP})$ with~$(u,v)\in
E_{\GAP}$ if and only if the item corresponding to~$u\in V$ is
placed into a bin corresponding to~$v\in V$.  A connected component
in~$G_{\GAP}$ with~$n'$
vertices has at most~$n'$ edges since every item can be
placed into at most one bin. If~$n' = 2$, we arbitrarily make one
of the vertices a bin and the other an item.
If~$n' > 2$, the connected component is a 1-tree, that is, a tree
and an edge. In this case, we partition the edges into two subgraphs;
a star forest and the disjoint union of a star forest and a cycle; see
Fig.~\ref{fig:1tree}.
Note that both subgraphs can be represented by touching boxes if we allow
point contacts.  This is due to the fact that the stars correspond to
a solution of~\prob{GAP}.
Hence, choosing a subgraph with larger weight and post-processing the
solution as in the proof of Theorem~\ref{thm:approx-bipartite-weighted}
results in a feasible solution of~\crown with no point contacts.
Initially, we discarded at most half of the weight and the post-processing
keeps at least~$3/4$ of the weight, so $\alg \ge 3\alg_{\GAP}/8$.
Therefore,~$\alg \ge 3\opt/(40\alpha)$.
\end{proofWithWrapfig}

\section{The Unweighted Case}
\label{sub:unweighted}

In this section, we consider the unweighted \crown problem, that is,
all desired contacts have profit~$1$.
Thus, we want to maximize the number of edges of the input graph
realized by the contact representation. We present approximation
algorithms for different graph classes.
First, we give a 2-approximation for trees.  Then, we present
a PTAS for planar graphs of bounded degree.
Finally, we provide a $(5 + 16\alpha/3)$-approximation for general graphs.

\begin{theorem}
  \label{thm:2-approx-tree-unweighted}
  Unweighted \crown on trees admits a 2-approximation.
\end{theorem}

\begin{proof}
   Let $T$ be the input tree.
  We first decompose $T$ into edge-disjoint stars as follows.
  If $T$ has at most two vertices, then the decomposition is straight-forward.
  So, we assume \WLOG that $T$ has at least three vertices and is rooted at a non-leaf vertex.
  Let~$u$ be a vertex of~$T$ such that all its children, say
  $v_1,\dots,v_ k$, are leaf vertices. 
  If~$u$ is the root of~$T$, then the decomposition contains only one
  star centered at $u$. 
  Otherwise, denote by~$\pi$ the parent of~$u$ in~$T$, create a
  star~$S_u$ centered at~$u$ with edges
  $(u,\pi),(u,v_1),\dots,(u,v_k)$ and call the edge $(u,\pi)$ of~$S_u$
  the \emph{anchor edge} of~$S_u$.
  The removal of $u,v_1,\ldots,v_k$ from $T$ results in a new tree.
  Therefore, we can recursively apply the same procedure.
  The result is a decomposition of~$T$ into edge-disjoint stars
  covering all edges of~$T$.

  We next remove, for each star, its anchor edge from~$T$.
  We apply the PTAS of Theorem~\ref{thm:previous-improved}
  to the resulting star forest and claim that the result
  is a $2$-approximation for~$T$.
  To prove the claim, consider a star $S_u'$ of the new star forest,
	centered at $u$ with edges $(u,v_1),\dots,(u,v_k)$
  and let $\alg$ be the total number of contacts realized by
	the $(1 + \eps)$-approximation algorithm on $S_u'$.
  We consider the following two cases.
  \begin{enumerate}[label=(\alph*), noitemsep, topsep=0pt, parsep=0pt, partopsep=0pt]
    \item $1\le k\le 4$:
    Since it is always possible to realize four contacts of a star,
    $\alg \ge k$.
    Note that an optimal solution may realize at most~$k+1$ contacts
    (due to the absence of the anchor edge from~$S_u'$).
    Hence, our algorithm has approximation ratio~$(k+1)/k \le 2$.
    \item $k \ge 5$:
    Since it is always possible to realize four contacts of a star,
    we have $\alg \ge 4$.
    On the other hand, an optimal solution realizes at
    most~$(1+\eps)\alg+1$ contacts.
    Thus, the approximation ratio is
    $((1+\eps)\alg+1)/\alg \le (1+\eps) + 1/4 < 2$.
  \end{enumerate}
  The theorem follows from the fact that all edges of~$T$ are incident to
  the star centers.
\end{proof}

Next, we develop a PTAS for bounded-degree planar graphs.
Our construction needs two lemmas, the first of which
was shown by Barth et al.~\cite{bfklnopsuw-swcrh-LATIN14}.
\begin{lemma}[\cite{bfklnopsuw-swcrh-LATIN14}]
  \label{lem:lower-bound-opt-deg}
  If the input graph $G=(V,E)$ has maximum degree $\Delta$ then\\
  $\opt\geq 2|E|/(\Delta+1)$.
\end{lemma}
The second lemma provides an exponential-time exact algorithm for \crown.
\begin{lemma}\label{lem:exact-algorithm}
  There is an exact algorithm for unweighted \crown with running time
  $2^{O(n\log n)}$.
\end{lemma}
\begin{proof}
  Consider a placement which assigns a position~$[\ell_B,r_B]\times [b_B,t_B]$
  to every box $B$, with~$\ell_B+w(B)=r_B$ and~$b_B+h(B)=t_B$.
  For the $x$-axis, this gives a (possibly nonstrict) linear order on the
  values~$\ell_B$ and~$r_B$, where some might be equal. An order on the $y$-axis is implied
  similarly.  Together, these two orders fully determine the
  combinatorial structure of overlaps and contacts: for contact,
  two boxes must have a side of equal value and a side with overlap.
  
  The algorithm
  enumerates all possible combinations of two such orders using
  the representation sketched above. On a single axis, this is a permutation of $2n$ variables and,
  between every two variables adjacent in this permutation, whether
  they are equal or the second variable has strictly larger value.
  This representation demonstrates that the number of distinct orders in one dimension is bounded
  by~$O((2n)!\cdot2^{2n})$, which is $2^{O(n\log n)}$.
	The number of combinations of two such orders also satisfies this bound.
  
  For any  given pair of orders, it can be determined
  if they imply overlaps  
  and what the objective value is: count the number of profitable contacts.
  If there are no overlaps, the existence of an actual placement
  realizing the orders is tested using linear programming. As these
  tests run in polynomial time, an optimal placement can be found
  in~$2^{O(n\log n)}$ time.
\end{proof}

\begin{theorem}\label{thm:ptas-bounded-deg-planar}
  Unweighted \crown on planar graphs with maximum degree~$\Delta$
  admits a PTAS.  More specifically, for
  any~$\eps>0$ there is an~$(1+\eps)$-approximation
  algorithm with linear running time~$n2^{(\Delta/\eps)^{O(1)}}$.
\end{theorem}
\begin{proof}
  Let~$r$ be a parameter to be determined later.
  Frederickson~\cite{frederickson87shortestpaths-planar} showed that
  we can find a vertex set~$X\subseteq V$ (called \emph{$r$-division}) of
  size~$O(n/\sqrt{r})$ such that the following holds. The vertex
  set~$V \setminus X$ can be partitioned into~$n/r$ vertex
  sets~$V_1,\dots,V_{n/r}$ such that (i)~$|V_i|\leq r$ for
  $i=1,\dots,n/r$ and (ii)~there is no edge running between any two
  distinct vertex sets~$V_i$ and~$V_j$.
  In what follows, we assume w.l.o.g.\ that~$G$ is connected, as we can
  apply the PTAS to every connected component separately.

  We apply the result of Frederickson to the input graph and compute an
  $r$-division~$X$.  By removing the vertex set $X$ from the graph, we
  remove $O(n\Delta/\sqrt{r})$ edges from $G$.  Now, we apply
  the exact algorithm of Lemma~\ref{lem:exact-algorithm} to each of
  the induced subgraphs $G[V_i]$ separately.  The solution is the
  union of the optimum solutions to $G[V_i]$.

  Since no edge runs between the distinct sets~$V_i$ and~$V_j$, the subgraphs
  $G[V_i]$ cover $G-X$.  Let $E^\star$ be the set of edges realized by an
  optimum solution to~$G$, let $\opt=|E^\star|$, and let
  $\opt'=|E^\star \cap E(G-X)|$.  By
  Lemma~\ref{lem:lower-bound-opt-deg}, we have that $\opt\geq
  2(n-1)/(\Delta+1)=\Omega(n/\Delta)$.  When we removed $X$ from $G$,
  we removed $O(n\Delta/\sqrt{r})$ edges.  Hence,
  $\opt=\opt'+O(n\Delta/\sqrt{r})$ and
  $\opt'=\Omega(n(1/\Delta-\Delta/\sqrt{r}))$.

  Since we solved each sub-instance $G[V_i]$ optimally and since these
  sub-instances cover $G-X$, the solution created by our algorithm
  realizes at least $\opt'$ many edges.  Using this fact and the above
  bounds on $\opt$ and $\opt'$, the total performance of our algorithm
  can be bounded by
  \begin{displaymath}
    \frac{\opt}{\opt'} \;=\; \frac{\opt'+O(n\Delta/\sqrt{r})}{\opt'} \;=\;
    1+O\left(\frac{n\Delta/\sqrt{r}}{n(1/\Delta-\Delta/\sqrt{r})}\right)
    \;=\; 1+O\left(\frac{\Delta^2}{\sqrt{r}-\Delta^2}\right)\,.
  \end{displaymath}
  We want this last term to be smaller than $1+\eps$ for some
  prescribed error parameter $0<\eps\leq 1$.  It is not hard to
  verify that this can be achieved by letting
  $r=\Theta(\Delta^4/\eps^2)$.  Since each of the subgraphs
  $G[V_i]$ has at most $r$ vertices, the total running time for
  determining the solution is $n2^{(\Delta/\eps)^{O(1)}}$.
\end{proof}

Before tackling the case of general graphs, we need a lower bound on
the size of maximum matchings in planar graphs in terms of the numbers
of vertices and edges.

\begin{lemma}
  \label{lem:matching-dense-planar}
  Any planar graph with $n$ vertices and $m$ edges
  contains a matching of size at least $(m-2n)/3$.
\end{lemma}

\begin{proof}
  Let $G$ be a planar graph.  Our proof is by induction on~$n$.  The
  claim clearly holds for~$n=1$.

  For the inductive step assume that $n>1$.  If $G$ is not connected,
  the claim follows by applying the inductive hypothesis to every
  connected component.  Now assume that $G$ has a vertex $u$ of degree
  less than~3.  Consider the graph $G'=G-u$ with $n'=n-1$ vertices and
  $m'\geq m-2$ edges.  By the inductive hypothesis $G'$ (and hence,
  $G$, too) has a matching of size at least
  \begin{displaymath}
		(m'-2n')/3 \ge ((m-2)-2(n-1))/3=(m-2n)/3.
  \end{displaymath}

  It remains to tackle the case where $G$ is connected and has minimum
  degree~3.  Nishizeki and Baybars~\cite{nishizeki79matchings-planar}
  showed that any connected planar graph with at least $n\geq 10$ vertices
  and minimum degree~3 has a matching of size at least
  $\lceil(n+2)/3\rceil\geq n/3$.  This shows the claim for $n\geq 10$
  since $m \le 3n-6$.
  
  In the remaining cases, $G$ has $n \le 9$ vertices. Due to planarity, we have
  $(m-2n)/3 \le (n-6)/3 \le 1$. Hence, any nonempty matching is large enough.
\end{proof}

We are now ready to present an approximation algorithm for general graphs.

\begin{theorem}
  \label{thm:approx-gen-graph-unweighted}
  Unweighted \crown on general graphs admits a $(5+16\alpha/3)(\approx13.4)$\--ap\-prox\-imation.
\end{theorem}

\begin{figure}[t]
  \begin{subfigure}[t]{.4\textwidth}
    \centering
    \includegraphics{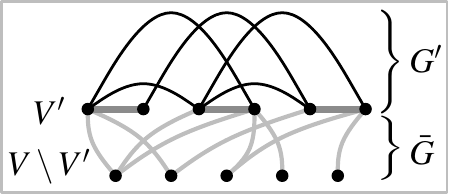}
    \caption{$G$ is covered by $\bar{G}$ (bipartite, gray) and $G'$;
      perfect matching $M$ (gray, bold).}
    \label{fig:matching1}
  \end{subfigure}
  \hfill
  \begin{subfigure}[t]{.21\textwidth}
    \centering
    \includegraphics{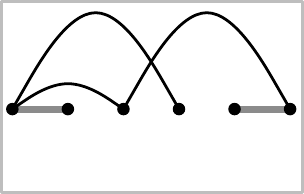}
    \caption{maximum matching $M''$ (gray/black) in $G''=G'{-}M$.}
    \label{fig:matching2}
  \end{subfigure}
  \hfill
  \begin{subfigure}[t]{.28\textwidth}
    \centering
    \includegraphics{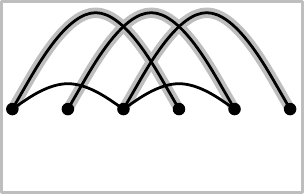}
    \caption{optimum solution to~$G'$: graph $G^*$
      (black) and part of~$M$ (gray).}
    \label{fig:matching3}
  \end{subfigure}
  \caption{Partitioning the input graph and the optimum solution in
    the proof of Theorem~\ref{thm:approx-gen-graph-unweighted}}
  \label{fig:matching}
\end{figure}

\begin{proof}
	The algorithm first computes a maximal matching $M$ in $G$. 
	Let $V'$ be the set of vertices matched by $M$, let $G'$ be the 
	subgraph induced by~$V'$, and let~$E'$ be the edge set of~$G'$.  
	Note that $ \bar{G}=G-E'$ is a bipartite graph with 
	partition $(V',V \setminus V')$. This is because the matching $M$ 
	is maximal, which implies that every edge in $E\setminus E'$ is 
	incident to a vertex in~$V'$ and to a vertex not in~$V'$; see
	Fig.~\ref{fig:matching1}. Hence, we can compute a $16\alpha/3$- 
	approximation to $\bar{G}$ using the algorithm presented in
	Theorem~\ref{thm:approx-bipartite-weighted}.
 
  Consider the graph $G''=(V',E'\setminus M)$ and compute a maximum
  matching~$M''$ in~$G''$; see Fig.~\ref{fig:matching2}.
  The edge set $M\cup M''$ is a set of
  vertex-disjoint paths and cycles and can therefore be completely
  realized~\cite{bfklnopsuw-swcrh-LATIN14}.  The
  algorithm realizes this set.  Below, we argue that
  this realization is in fact a $5$-approximation for $G'$, which completes
  the proof (due to Lemma~\ref{lem:partitioned-approximation} and
  since $G$ is covered by~$G'$ and~$\bar{G}$).

  Let $n'=|V'|$ be the number of vertices of~$G'$.  Let $E^*$ be the set
  of edges realized by an optimum solution to~$G'$, and let $\opt=|E^*|$.
  Consider the subgraph $G^*=(V',E^*\setminus M)$ of~$G''$;
  see Fig.~\ref{fig:matching3}.
  Note that $G^*$ is planar and contains at least $\opt-n'/2$
  many edges.  Applying Lemma~\ref{lem:matching-dense-planar}
  to~$G^*$, we conclude that the maximum matching~$M''$ of~$G''$ has
  size at least $(\opt-5n'/2)/3$.  Hence, by splitting~$\opt$
  appropriately, we obtain
  \begin{displaymath}
   \opt \;=\; (\opt-5n'/2)+5n'/2 \;\le\; 3|M''|+5|M| \;\le\; 5|M''
   \cup M|\,. \qedhere
  \end{displaymath}
\end{proof}

\section{The Model with Point Contacts}
\label{sec:point-model}

In the model with point contacts, adjacencies between boxes may be
realized by a \emph{point contact}, that is, if two boxes
touch each other in two corners. Note that the algorithms that use
the PTAS of Lemma~\ref{lem:gap-ptas}
also hold for this model without any modification.

\subsection{Weighted bipartite and general graphs.}
For these graph classes, we do, on the one hand, no longer need the
post-processing that we applied in
Theorems~\ref{thm:approx-bipartite-weighted}
and~\ref{thm:approx-gen-graph-weighted} (and implicitly also in
Theorem~\ref{thm:randomized-approx-gen-graph-weighted}).  This
post-processing cost us up to a quarter of the total profit.  Hence,
we can (for now) replace $\alpha$ by $3\alpha/4$, which improves the
approximation factors for these cases.

On the other hand, a realized graph is now not necessarily planar as
four boxes can meet in a point and both diagonals correspond to
edges of the input graph.  It is, however, easy to see that the graphs
that can be realized are 1-planar.  This means that an optimal
solution has at most $4n-8$ edges in the case of general graphs and at
most $3n-6$ edges in the case of bipartite graphs.  Furthermore,
Ackerman~\cite{ackerman2014notejournal}
showed very recently that a 1-planar graph can be covered by a planar
graph and a tree.  Hence, we can cover a 1-planar graph with seven
star forests and a bipartite 1-planar graph with six star forests (via
a bipartite planar graph and a tree).

If our approximation algorithm for bipartite graphs uses this
decomposition into six star forests, we easily get a
$6\alpha$-approximation for this case.  As a consequence, we get (as
in Theorem~\ref{thm:randomized-approx-gen-graph-weighted}) a
randomized $12\alpha$-approximation for general graphs.  Similarly,
decomposing an optimum 1-planar solution into seven star forests
(instead of five star forests for planar graphs), we get a
deterministic $14\alpha$-approximation for general graphs.

\begin{theorem}
  \label{thm:point-contacts-weighted}
  Weighted \crown in the model with point contacts admits
  a~$6\alpha(\approx9.5)$-approximation algorithm on bipartite graphs,
  a randomized~$12\alpha(\approx19)$-approximation algorithm on general
	graphs, and a deterministic~$14\alpha(\approx22.1)$-approximation
	algorithm on general graphs.
\end{theorem}

\subsection{Unweighted general graphs.}
In order to modify the algorithm for the unweighted case, we use the new
decomposition of bipartite graphs. It is easy to prove that any
1-planar graph with $m$ edges and $n$ vertices contains a matching of
size at least~$(m-3n)/3$: we planarize the graph (by removing at most
$n$ edges) and then apply Lemma~\ref{lem:matching-dense-planar}.  This
results in a $(7+6\alpha)$-approximation for unweighted general graphs.

\begin{theorem}
  \label{thm:point-contacts-unweighted}
  Weighted \crown in the model with point contacts admits
  a~$(7+6\alpha)(\approx16.5)$-approximation algorithm on unweighted
	general graphs.
\end{theorem}

\section{APX-Completeness}
\label{sub:hardness}

In this section, we prove APX-completeness of weighted \crown
by giving a reduction from 3-dimensional matching. This reduction
works both in the model without and in the model with point contacts.

\begin{theorem}
  \label{thm:bipartite-weighted-apx-complete}
  Weighted \crown is APX-complete even if the input graph is bipartite of
  maximum degree~9, each edge has profit~1,~2 or~3, and each vertex
  corresponds to a square of one out of three different sizes.
\end{theorem}
\begin{proof}
 We give a reduction from 3-dimensional matching (3DM). An instance
  of this problem is given by three disjoint sets $X,Y,Z$ with
  cardinalities $|X|=|Y|=|Z|=k$ and a set $E\subseteq X\times Y\times Z$
  of hyperedges.  The objective is to find a set $M\subseteq E$,
  called \emph{matching}, such that no element of $V=X\cup Y\cup Z$
  is contained in more than one hyperedge in $M$ and such that $|M|$
  is maximized.

  The problem is known to be APX-hard~\cite{Fleischer2011}.  More
  specifically, for the special case of 3DM where every $v\in V$ is
  contained in at most three hyperedges (hence $|E|\leq 3k$) it is
  NP-hard to decide whether the maximum matching has cardinality~$k$
  or only $k(1-\eps_0)$ for some constant $0<\eps_0<1$.  We reduce
  from this special case of 3DM to \crown.

  To this end, we construct the following \crown instance from a
  given 3DM instance.  We create, for each $v\in V$, a square of side
  length~1.  For each hyperedge $e\in E$, we create nine squares
  $e^\star,e_1,\dots,e_8$ where~$e^\star$ has side length~3.5 and
  $e_1,\dots,e_8$ have side length~3.  In the desired contact graph,
  we create an edge $(e^\star,e_1)$ of profit~2 and, for
  $i=2,\dots,8$, an edge $(e^\star,e_i)$ of profit~3.  We also create
  an edge $(e^\star,v)$ of profit~1 if~$v$ is incident to~$e$ in the
  3DM instance.

  Consider an optimum solution to the above \crown instance.  It is
  not hard to verify that, for any hyperedge $e=(x,y,z)$, the solution
  will realize the edges $(e^\star,e_i)$ for $i=2,\dots,8$.  Moreover,
  we can assume w.l.o.g.\ that the solution either realizes all three
  adjacencies $(e^\star,x)$, $(e^\star,y)$, and $(e^\star,z)$ of total
  profit~3 or the adjacency $(e^\star,e_1)$ of profit~2; see
  Fig.~\ref{fig:illustration-apx-completeness}.  We call such a
  solution \emph{well-formed}.

	\begin{figure}[t]
		\centering
		\begin{minipage}[b]{.6\textwidth}
			\begin{subfigure}[t]{.47\textwidth}
				\centering
				\includegraphics[page=1]{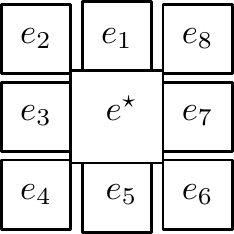}
				\caption{profit $7 \cdot 3 + 2 = 23$}
			\end{subfigure}
			\hfill
			\begin{subfigure}[t]{.47\textwidth}
				\centering
				\includegraphics[page=2]{apx-hardness}
				\caption{profit $7 \cdot 3 + 3 \cdot 1 = 24$}
			\end{subfigure}
		\end{minipage}
		\caption{The two possible configurations of a hyperedge $e=(x,y,z)$
			in the proof of Theorem~\ref{thm:bipartite-weighted-apx-complete}}
		\label{fig:illustration-apx-completeness}
	\end{figure}

  Assume that there is a solution $M$ to the 3DM instance of
  cardinality $k$.  Then this can be transformed into a well-formed
  solution to \crown of profit $(7\cdot 3+2)|E|+|M|=23|E|+k$.

  Conversely, suppose that the maximum matching has cardinality at most
  $(1-\eps_0)k$.  Consider an optimum solution to the respective
  \crown instance.  We may assume that the solution is
  well-formed.  Let $M$ be the set of hyperedges $e=(x,y,z)$ for which
  all three adjacencies $(e^\star,x), (e^\star,y), (e^\star,z)$ are
  realized.  Then,
  the profit of this solution is $(7\cdot 3+2)|E|+|M|=23|E|+|M|$.  Note
  that $M$ is in fact a matching because the solution to \crown
  was well-formed.  Thus, the optimum profit is bounded by
  $23|E|+(1-\eps_0)k$.

  Hence, it is NP-hard to distinguish between instances
  with~$\opt\geq 23|E|+k$ and instances with~$\opt\leq 23|E|+(1-\eps_0)k$.
  Using $|E|\leq 3k$, this implies that there cannot be any
  approximation algorithm of ratio less than
  \begin{displaymath}
    \frac{23|E|+k}{23|E|+(1-\eps_0)k} \;=\;
    1+\frac{\eps_0k}{23|E|+(1-\eps_0)k} \;\ge\;
    1+\frac{\eps_0k}{(70-\eps_0)k} \;=\; 1+\frac{\eps_0}{70-\eps_0}
    \,,
  \end{displaymath}
  which is a constant strictly larger than~$1$.
\end{proof}

\section{Conclusions and Open Problems}
\label{sec:open}

We presented approximation algorithms for the \crown problem, which can
be used for constructing semantics-preserving word clouds. Apart from
improving approximation factors for various graph classes, many
open problems remain. Most of our algorithms
are based on covering the input graph by subgraphs and
packing solutions for the individual subgraphs. Both
subproblems---covering graphs with special types of subgraphs and
packing individual solutions together---are interesting problems in
their own right which may lead to algorithms with better guarantees.
Practical variants of the problem are also of interest, for
example, restricting the heights of the boxes to predefined values
(determined by font sizes), or defining more than immediate neighbors
to be in contact, thus considering non-planar ``contact'' graphs.
Another interesting variant is when the bounding box of the
representation has a certain fixed size or aspect ratio.

\bibliographystyle{abbrv}
\bibliography{abbrv,wordle}

\end{document}